\documentclass[conference,10pt]{IEEEtran}

\usepackage{url}
\usepackage[pdftex]{graphicx}
\usepackage[cmex10]{amsmath}
\usepackage{amssymb}
\usepackage{amsthm}
\usepackage{float}
\usepackage{epstopdf}
\usepackage{hyperref}
\usepackage{ dsfont,color }
\usepackage{cite}
\usepackage{graphics}

\newtheorem{Theorem}{Theorem}

\ifCLASSINFOpdf
\fi
\begin{document}

\title{Achievable Degrees of Freedom of the K-user MISO Broadcast Channel with Alternating CSIT via Interference Creation-Resurrection \vspace{-1em}}
\author{\IEEEauthorblockN{Mohamed Seif\IEEEauthorrefmark{1}, Amr El-Keyi\IEEEauthorrefmark{3}, and Mohammed Nafie\IEEEauthorrefmark{1}\IEEEauthorrefmark{2}}\\

\IEEEauthorblockA{\IEEEauthorrefmark{1}Wireless Intelligent Networks Center (WINC), Nile University, Egypt\\
\IEEEauthorrefmark{2}Dept. of EECE, Faculty of Engineering, Cairo University, Egypt\\
\IEEEauthorrefmark{3}Department of Systems and Computer Engineering, Carleton University, Ottawa, Canada\\
        Email: {m.seif@nu.edu.eg,  amr.elkeyi@sce.carleton.ca, mnafie@nu.edu.eg}
}}


\maketitle

\begin{abstract}

Channel state information at the transmitter affects the degrees of freedom of the wireless networks. In this paper, we analyze the DoF for the K-user
multiple-input single-output (MISO) broadcast channel (BC) with
synergistic alternating  channel state information at the
transmitter (CSIT). Specifically, the CSIT of each user alternates between
three states, namely, perfect CSIT (P), delayed CSIT (D) and no
CSIT (N) among different time slots. For the K-user MISO BC, we show that the total achievable degrees of freedom (DoF) are given by $\frac{K^{2}}{2K-1}$ through utilizing the synergistic benefits of CSIT patterns. We compare the achievable DoF with results reported previously in the literature in the case of delayed CSIT and hybrid CSIT models.
\end{abstract}

\IEEEpeerreviewmaketitle





\smallskip
\noindent \textbf{Index Terms:} Broadcast channel, degrees of
freedom, interference alignment, alternating CSIT, interference
creation-resurrection.

%
\IEEEpeerreviewmaketitle





\makeatletter{\renewcommand*{\@makefnmark}{} \footnotetext{ \vspace{0.02in} 
 This work is supported by a grant from the Egyptian NTRA. The statements made herein are solely the responsibility of the author[s].”
}\makeatother}

\section{Introduction}

Due to the rapid growth in wireless traffic, interference
management is essential to provide the required quality
of service (QoS) for future wireless networks. Traditional prior
work focused on reducing the interference power at the receivers.
Recently, interference alignment (IA) has been proposed and
studied on various networks such as the interference, broadcast
and X channels. IA is an elegant technique to decrease the impact
of interference through reducing the dimension of the
interference subspace thanks to the seminal work of
\cite{maddah2008communication,cadambe2008interference}.


An important performance measure for a communication network is
its degrees of freedom (DoF) which determines the behavior of the
sum capacity in the high signal-to-noise ratio (SNR) regime. In
particular, the network capacity under a transmission power $P$ is
given by  \cite{jafar2011interference}
\begin{equation}
C(P)=\text{DoF} \log(P)+o(\log(P))
\end{equation}
where $\lim_{P\rightarrow \infty} \frac{o(\log(P))}{\log(P)} = 0$.


In capacity characterization work, it is a common assumption that
receivers know the channel state information (CSI) perfectly and
instantaneously, while the CSI knowledge at the transmitter(s)
(CSIT) is usually subject to some limitations. At one extreme, it
is assumed that the transmitters know the CSI instantaneously and
perfectly (full CSIT assumption). Under this condition, the
capacity region and, hence the DoF region, of the multiple-input
multiple-output (MIMO) broadcast channel was characterized in
\cite{weingarten2006capacity}. The DoF of the K-user single-input
single-output (SISO) interference channel was shown to be
$\frac{K}{2}$ with full CSIT \cite{cadambe2008interference}. Also,
it was shown in \cite{cadambe2009interference} that the $M \times
K$ SISO $X$ channel with full CSIT has $\frac{MK}{M + K - 1}$ DoF.
In \cite{cadambe2009degrees}, it was proved that channel output
feedback does not provide any DoF benefit in interference and X
channels under the full CSIT assumption. At the other extreme, the
transmitter(s) are assumed to have no knowledge about CSI. In this
case, the K-user multiple-input single-output (MISO) broadcast
channel was studied in \cite{jafar2005isotropic}. Other works
include \cite{vaze2012degree} which characterized the DoF regions
of the K-user MIMO broadcast channel, interference channel and X
channel. Also, \cite{huang2012degrees,zhu2012degrees,vaze2012new}
studied the DoF region of the two-user MIMO broadcast and
interference channels with no CSIT by developing upper and lower
bounds on the DoF. It was shown in \cite{vaze2012degree} that the
MISO broadcast, SISO interference and SISO X channels under
isotropic i.i.d. fading can achieve no more than one DoF.

Maddah Ali and Tse investigated a delayed CSIT model, which is an
intermediate assumption between the two extremes; full CSIT and no
CSIT. This model was introduced in \cite{maddah2012completely} for
the K-user Gaussian MISO broadcast channel (BC). They showed that
the K-user BC under delayed CSIT can achieve at most
$K/(1+\frac{1}{2}+\dots+\frac{1}{K})$ DoF which is strictly
greater than one DoF. Also, in \cite{maleki2012retrospective},
Maleki et al. applied the delayed CSIT model to the X-channel and
showed that the 2 user SISO X channel under delayed CSIT
assumption can achieve $\frac{8}{7}$ DoF. A variety of work
concerning CSIT availability models have been studied such as:
quantized CSIT \cite{jindal2006mimo,kobayashi2008much}, compound
CSIT
\cite{weingarten2007compound,gou2011degrees,maddah2010degrees} and
mixed CSIT \cite{gou2012optimal}. \vspace{-1mm}
\subsection*{Related Work}
Another interesting model is the alternating CSIT model that was
first introduced by Tandon et. al. in
\cite{tandon2013synergistic}. The authors of the pre-mentioned
paper studied the synergistic benefits of alternating CSIT for the
2-user MISO broadcast channel and defined the DoF region
$\mathcal{D}$ for different patterns of alteration. Also, the same
authors in \cite{tandon2012minimum} studied the K-user case and
identified the minimum CSIT pattern to achieve the upper bound on
the total DoF,  which is given by $\min (M,K)$, for the MISO
broadcast channel with an $M$ antenna transmitter and $K$ single
antenna users. The achievable DoF under this model is upper
bounded by
\begin{equation}
D_{\Sigma}(K) \leq \frac{K(M+(\min(M,K)-1)\lambda)}{M+K-1}
\end{equation}
where $\lambda=\frac{\min(M,K)}{K}$ is the fraction time that CSIT
is perfect per user.

In \cite{amuru2014degrees}, the authors considered the hybrid CSIT
model for the BC; in which the CSIT pattern is fixed during the
channel uses. In their framework, there is a perfect CSIT for a
subset of receivers and delayed CSIT for the remaining receivers.
For the 3-user case, they showed that for a 2-antenna transmitter
with perfect CSIT for one user and delayed CSIT for the other two
users, the BC can achieve at most $\frac{5}{3}$ DoF. Also, they
studied the system with 3 antennas at the transmitter and showed
that for the previous hybrid CSIT pattern a total DoF of
$\frac{9}{5}$ is achievable. For the same number of antennas, i.e.
three, but with a higher CSIT setting; in which perfect CSIT is
available for two users while delayed CSIT for the third user,
$\frac{9}{4}$ total DoF can be achieved.

The authors of \cite{wagdy2014degrees} studied the SISO X channel with
synergistic alternating CSIT. They proposed schemes based on
interference creation-resurrection (ICR) that achieve the upper
bound on the DoF of the 2-user network which is $\frac{4}{3}$ DoF. Also, they
characterized the DoF region $\mathcal{D}$ as a function of the
distribution of CSIT states, that are basically; perfect $(P)$,
delayed $(D)$ and no CSIT $(N)$.

In this paper, we propose a scheme based on ICR under alternating
CSIT for the K-user BC. The ICR scheme is partitioned into two
phases: phase one is associated with the delayed CSIT and no CSIT states. In this phase, 
information terms are delivered to receivers with no CSIT availability  and
interference terms (to be resurrected in phase two) are received by receivers with delayed CSIT. In phase two, we deliver useful linear
combinations of past interference terms to the receivers in order
to decode their desired messages. We show that the achievable DoF for this network is given by
\begin{equation}
 D_{\Sigma}(K) = \frac{K^{2}}{2K-1}
\end{equation}
\noindent and the distribution of fraction of time of the different states $\{P,D,N\}$ required for our proposed scheme is
\begin{equation}
\lambda_{P} = \frac{(K-1)^{2}}{2K^{2}-K}, \lambda_{D} = \frac{K-1}{2K-1}, \lambda_{N}=\frac{1}{K}.
\end{equation}



The rest of the paper is organized as follows. Section II
describes the system model. The proposed scheme is discussed in
Section III. Section IV provides numerical evaluation of the
attained  DoF expression and shows the performance gains for our
proposed system compared to previous work. Finally, we conclude
the paper in Section V.


\section{System Model}
We consider a MISO broadcast channel with $K$ transmit antennas
and $K$ single antenna receivers. The received signal at the $i$th
receiver is given by
\begin{equation}
Y_{i}(t)=H_{i}(t) X(t)+N_{i}(t), \hspace{0.1in} i=1, \dots, K
\end{equation}
where $X(t)$ is the $K \times 1$ transmitted signal at time $t$
with a power constraint $E\{|X(t)|^{2}\} \leq P$. The additive
noise $N_{i}(t)\sim\mathcal{CN}(0,1)$ at time $t$ generated at
receiver $R_{i}$ is circularly symmetric white Gaussian noise with
zero mean and unit variance. $H_{i}(t)$ is the $1 \times K$
channel vector from the transmitter to receiver $R_{i}$ at time
$t$ which is sampled from a continuous distribution whose elements
are complex Gaussian. The channel coefficients are assumed to be
i.i.d. across the receivers. Let $r_{i}(P)$ denote the achievable
rate of message $W_{i}$ for a given transmission power $P$ defined
as $r_{i}(P)=\frac{\log_{2}(|W_{i}|)}{n}$ where $|W_{i}|$ is the
cardinality of the message set and $n$ is the number of channel
uses. The DoF region $\mathcal{D}$ is defined as the set of all
achievable tuples $(d_{1},d_{2}, \dots, d_{K}) \in
\mathds{R}_{+}^{K}$ where $d_{i}=\lim_{P \rightarrow \infty}$
$\frac{r_{i}(P)}{\log_{2}(P)}$ is the DoF for message $W_{i}$. The
total DoF of the network is defined as
\begin{equation}
D_{\Sigma}(K)= \max _{ (d_{1}, d_{2}, \dots, d_{K}) \in
\mathcal{D}}{ d_{1}+ d_{2}+\dots+d_{K} }.
\end{equation}

     \begin{figure}[ht!]
        \centering
        \includegraphics[scale=0.35]{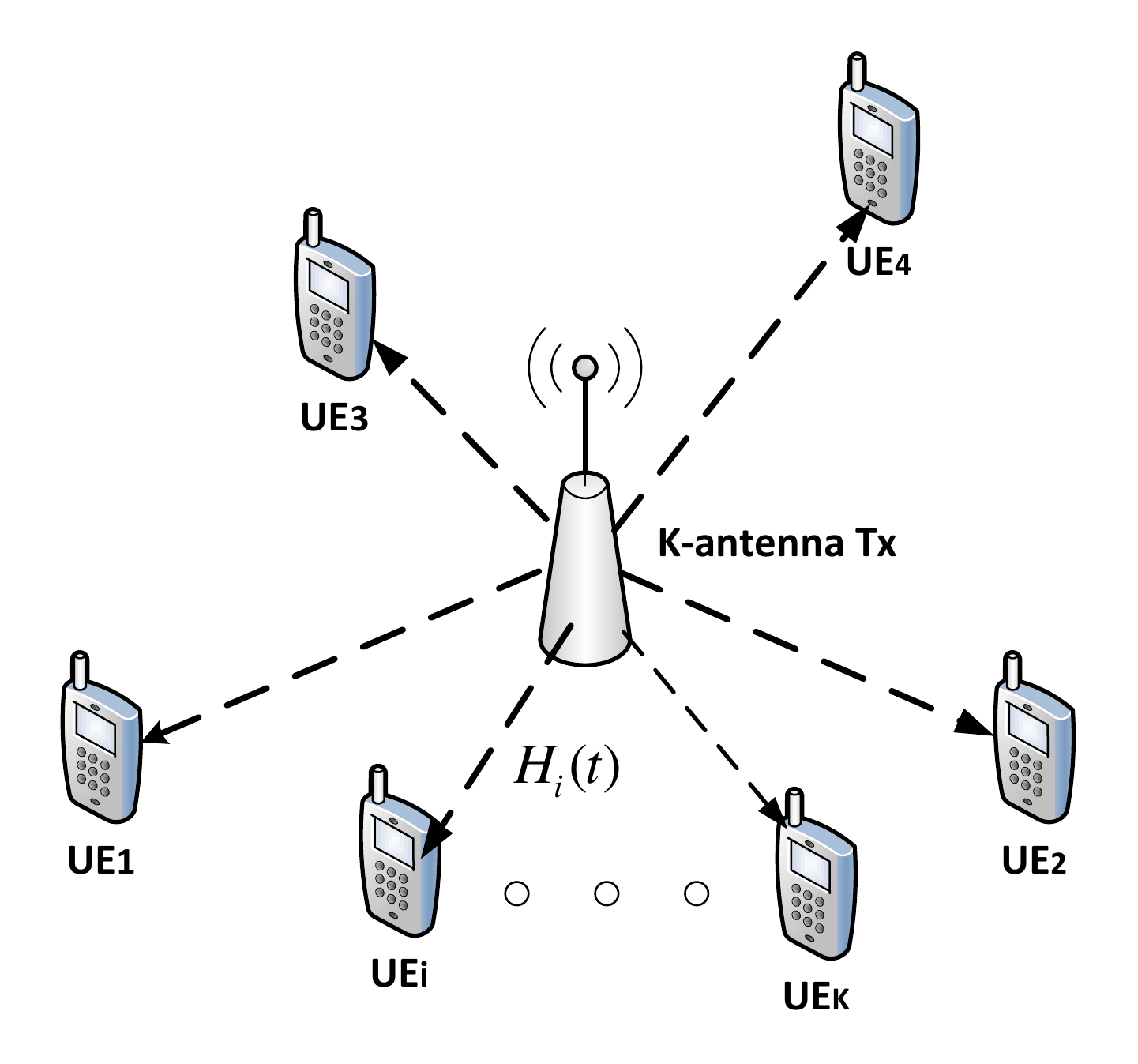}
        \caption{Network Model: A MISO BC with a K-antenna transmitter and K single antenna users. }
        \label{fig::BC_system_model}

     \end{figure}
We assume that the receivers have perfect and global channel state
information. Furthermore, we consider three different states of
the availability of CSIT
\begin{enumerate}
    \item Perfect CSIT ($P$): identifies the state of CSIT in which CSIT is available to the transmitter instantaneously and without error.
    \item Delayed CSIT ($D$): identifies the state of CSIT in which CSIT is available to the transmitter with some delay greater than or equal one time slot duration and without error.
    \item No CSIT ($N$): identifies the state of CSIT in which CSIT is not available to transmitter at all.
\end{enumerate}
The state of CSIT availability of the channel to the $i$th
receiver at time instant $t$ is denoted by $S_{i}(t)$; where, $S_{i}(t)\in \{P,D,N\}$. For
instance, $S_{2}(t)=P$ indicates that the transmitter has perfect and
instantaneous knowledge of $H_{2}$ at time instant $t$. In addition, let $S_{12\dots
K}(t)$ denote the collection of the states of CSIT availability of the channels to the
receivers $\{1,2,\dots, K\}$ at time slot $t$, respectively. Therefore,
$S_{12\dots K}(t)\in\{PP\dots P,PP\dots D,\dots, NN\dots N\}$. For
example, $S_{123}(t)=PDN$, refers to the case where the transmitter
has perfect knowledge to $H_{1}$, delayed information about
$H_{2}$ and no information about $H_{3}$. We denote the CSIT availability of the channels to the $i$th
receiver over $n$ time slots by $S_{i}^{n}$. For instance, the
CSIT availability over three time slots for receiver $R_{i}$ is
given by $S_{i}^{3}=(x,y,z)$ where $x,y,z\in S_{i}$ and $x,y$ and
$z$ denote the availability of CSIT in the first, second and third
time slots, respectively. Similarly, we denote the availability of
CSIT for the channels to the first and second receivers in three
time slots ``CSIT pattern''  by $S_{12}^{3}=(X,Y,Z)$ where
$X,Y,Z\in S_{12}$.

The fraction of time associated with the availability of CSIT state $S$ for the network, denoted by $\lambda_{S}$ where $S
\in \{P,D,N\}$, is given by
\begin{equation}
    \lambda_{S}= \frac{\sum_{t=1}^{n} \sum_{i=1}^{K}\mathds{I}(S_{i}(t)=S)}{nK}
\end{equation}

\noindent where 
\begin{equation}
\mathds{I}(S_{i}(t)=S) = \begin{cases}
1, &  \text{if $S_{i}(t)=S$}  \\
0, & \text{otherwise}
\end{cases}
\end{equation}
\noindent and $n$ is the number of channel uses, and hence,
\begin{equation}
    \sum_{S \in \{P,D,N\}}^{} \lambda_{S}=1.
\end{equation}

Furthermore, we use $\Lambda(\lambda_{P},\lambda_{D},\lambda_{N})$
to denote the distribution of fraction of time of the different
states $\{P,D,N\}$ of CSIT availability.


\section{Proposed Interference Creation-Resurrection Scheme}
Motivated by the previous work of \cite{wagdy2014degrees} for the
X channel, we extend this work to the BC. In this section, we
propose a precoding scheme for the BC under alternating CSIT. The
scheme is divided into two phases. The first phase is associated
with the delayed and no CSIT states where the transmitter sends its
messages. As a result, the receivers get linear combinations of
their desired messages in addition to interference terms during
this phase. This phase is called ``interference creation.'' On the
other hand, the second phase is associated with the perfect CSIT state
and is called ``interference resurrection'' phase. In this phase,
the transmitter reconstructs the old interference by exploiting
the delayed CSIT in phase one in order to deliver new linear
combinations to the receivers free from the interference and
enable the receivers to extract their desired messages via
physical network coding.

As an illustrative example of the K-user case: first, we consider
a 3-user MISO BC with alternating CSIT pattern given by $S_{123}^{5} =
(NDD,DND,DDN,PPN,PNP)$ over five time slots. Let $u_{1}, u_{2}$
and $u_{3}$ be three independent messages intended to receiver
$R_{1}$, $v_{1}, v_{2}$ and $v_{3}$ be three independent messages
intended to receiver $R_{2}$, and $p_{1}, p_{2}$ and $p_{3}$ be
three independent messages intended to receiver  $R_{3}$.
Consequently, the proposed scheme is performed over two phases as
follows in the next subsections.
\vspace{-1em}
\subsection{Phase 1: \textbf{Interference Creation}}

 This phase consists of three time slots, each time slot is intended to
deliver an interference-free linear combination of the messages
intended for one receiver. Therefore, at the $i$th time slot,
$R_{i}$ receives a linear combination of its desired symbols while
the two other receivers $R_{j}, j\in\{1,2,3\}\backslash\{i\}$
receive interference terms.

At $t = 1$:\\
The transmitter sends all data symbols for $R_{1}$, i.e.,
\begin{equation}
X(1)= \left[ \begin{matrix} { u }_{ 1 } \\ { u }_{ 2 } \\ { u }_{ 3 } \end{matrix} \right].
\end{equation}

\noindent As a result, the received signals are given as:
\begin{eqnarray}
Y_{1}(1) &=& H_{1}(1) X(1)= L_{1}^{1}(u_{1},u_{2},u_{3})\\
Y_{2}(1) &=& H_{2}(1)X(1) = I_{2}^{1}(u_{1},u_{2},u_{3})\\
Y_{3}(1) &=& H_{3}(1) X(1) = I_{3}^{1}(u_{1},u_{2},u_{3})
\end{eqnarray}
where $L_{i}^{j}(x_{1},x_{2},x_{3})$ denotes the $j$th linear
combination of the messages $x_{1}, x_{2}$ and $x_{3}$ that is
intended for receiver $R_{i}$ and $I_{i}^{j}(z_{1},z_{2},z_{3})$
denotes the $j$th interference term for receiver $R_{i}$ which is
a function of the messages $z_{1},z_{2}$ and $z_{3}$ overheard by
receiver $R_{i}$.

At $t = 2$:\\
Similarly, the transmitter sends all data symbols for $R_{2}$ as follows:
\begin{equation}
X(2)= \left[ \begin{matrix} { v }_{ 1 } \\ { v }_{ 2 } \\ { v }_{ 3 } \end{matrix} \right].
\end{equation}
Then, the received signals are:
\begin{eqnarray}
Y_{1}(2) &=& H_{1}(2) X(2)= I_{1}^{1}(v_{1},v_{2},v_{3})\\
Y_{2}(2) &=& H_{2}(2)X(2)= L_{2}^{1}(v_{1},v_{2},v_{3})\\
Y_{3}(2) &=& H_{3}(2) X(2) = I_{3}^{2}(v_{1},v_{2},v_{3})
\end{eqnarray}


At $t = 3$:\\
Finally, the transmitter sends all data symbols for $R_{3}$:
\begin{equation}
X(3) = \left[ \begin{matrix} { p }_{ 1 } \\ { p }_{ 2 } \\ { p }_{ 3 } \end{matrix} \right].
\end{equation}
Then,
\begin{eqnarray}
Y_{1}(3) &=& H_{1}(3) X(3)= I_{1}^{2}(p_{1},p_{2},p_{3})\\
Y_{2}(3) &=& H_{2}(3)X(3) = I_{2}^{2}(p_{1},p_{2},p_{3})\\
Y_{3}(3) &=& H_{3}(3) X(3)= L_{3}^{1}(p_{1},p_{2},p_{3}).
\end{eqnarray}


\subsection{Phase 2: \textbf{Interference Resurrection}}

This phase consists of two time slots where in each time slot the
transmitted signal is designed such that it provides two
interference-free linear combinations of the messages intended to
two receivers while the third receiver gets a linear combination
of its desired messages corrupted by an interference term that can
be removed using the received interference in previous time slots.

At $t = 4$:\\
In this time slot, the transmitter utilizes the perfect CSIT at
$R_{1}$ and $R_{2}$. The transmitter delivers two interference-free terms to $R_{1}$ and $R_{2}$ while providing an interference-corrupted desired term
for $R_{3}$. The transmitted signal
is given by
\begin{eqnarray}
X(4) &=& h_{1}^{\perp}(4) \left[ \begin{matrix} { I }_{ 3 }^{2}(v_{1},v_{2},v_{3}) \\ 0 \\ 0 \end{matrix} \right] + h_{2}^{\perp}(4)\left[ \begin{matrix} { I }_{ 3 }^{1}(u_{1},u_{2},u_{3}) \\ 0 \\ 0 \end{matrix} \right] \nonumber\\
 &+& h_{(1,2)}^{\perp}(4)\left[ \begin{matrix} { I }_{ 1 }^{1}(p_{1},p_{2},p_{3}) \\ 0 \\ 0 \end{matrix} \right]
\end{eqnarray}
where $h_{i}(t)^{\perp}$ and $h_{(i,j)}(t)^{\perp} \in
\mathds{C}^{3 \times 3}$ are the orthogonal projection matrices on
the null space of $H_{i}(t)$ and on the null space of the subspace
spanned by both $H_{i}(t), H_{j}(t)$, respectively. Then,
\begin{eqnarray}
Y_{1}(4) &=& \left[H_{1}(4) h_{2}^{\perp}(4)\right]_{1}{ I }_{ 3 }^{1}(u_{1},u_{2},u_{3})   \\
&=& L_{1}^{2}(u_{1},u_{2},u_{3})\\
Y_{2}(4) &=& \left[H_{2}(4)h_{1}^{\perp}(4)\right]_{1}  { I }_{ 3 }^{2}(v_{1},v_{2},v_{3})  \\
 &=& L_{2}^{2}(v_{1},v_{2},v_{3})\\
Y_{3}(4)
&=& L_{3}^{2}(p_{1},p_{2},p_{3})+\left[H_{3}(4)h_{1}^{\perp}(4)\right]_{1}I_{3}^{2}(v_{1},v_{2},v_{3}) \nonumber\\
&+& \left[H_{3}(4)h_{2}^{\perp}(4)\right]_{1}I_{3}^{1}(u_{1},u_{2},u_{3})
\end{eqnarray}
where $\left[X\right]_{1}$ is the first element of a vector $X \in
\mathds{C}^{1 \times 3}$. In spite of receiving an interference-corrupted signal, receiver $R_{3}$ can get a linear combination of its desired signals only and remove the interference by applying a simple physical network coding as follows:
\begin{eqnarray}
L_{3}^{2}(p_{1},p_{2},p_{3})&=& Y_{3}(4)-\left[H_{3}(4)h_{1}^{\perp}(4)\right]_{1}Y_{3}(2) \nonumber\\
&-& \left[H_{3}(4)h_{2}^{\perp}(4)\right]_{1} Y_{3}(1)
\end{eqnarray}

At $t = 5$:\\
In this time slot, we deliver two interference-free terms to
$R_{1}$ and $R_{3}$ while providing a desired term for $R_{2}$
corrupted by removable interference, i.e.,
\begin{eqnarray}
X(5) &=& h_{1}^{\perp}(5) \left[ \begin{matrix} { I }_{ 2 }^{2}(p_{1},p_{2},p_{3}) \\ 0 \\ 0 \end{matrix} \right] + h_{3}^{\perp}(5) \left[ \begin{matrix} { I }_{ 2 }^{1}(u_{1},u_{2},u_{3}) \\ 0 \\ 0 \end{matrix} \right] \nonumber\\
&+& h_{(1,3)}^{\perp}(5) \left[ \begin{matrix} { I }_{ 1 }^{1}(v_{1},v_{2},v_{3}) \\ 0 \\ 0 \end{matrix} \right]
\end{eqnarray}
Then,
\begin{eqnarray}
Y_{1}(5) &=& \left[H_{1}(5)h_{3}^{\perp}(5)\right]_{1}  { I }_{ 2 }^{1}(u_{1},u_{2},u_{3})\\
 &=& L_{1}^{3}(u_{1},u_{2},u_{3})\\
Y_{2}(5)&=&  L_{2}^{3}(v_{1},v_{2},v_{3})
+ \left[H_{2}(5) h_{1}^{\perp}(5)\right]_{1} { I }_{ 2 }^{2}(p_{1},p_{2},p_{3}) \nonumber\\
&+&\left[H_{2}(5) h_{3}^{\perp}(5)\right]_{1}  { I }_{ 2 }^{1}(u_{1},u_{2},u_{3})  \\
Y_{3}(5)&=& \left[H_{3}(5) h_{1}^{\perp}(5)\right]_{1}  { I }_{ 2 }^{2}(p_{1},p_{2},p_{3}) \\
&=& L_{3}^{3}(p_{1},p_{2},p_{3})
\end{eqnarray}
Receiver $R_{2}$ can also remove the interference signal using its received signal in previous time slots, i.e.,
\begin{eqnarray}
L_{2}^{3}(v_{1},v_{2},v_{3})&=& Y_{2}(5)-\left[H_{2}(5)h_{3}^{\perp}(5)\right]_{1}Y_{2}(1) \nonumber\\
&-& \left[H_{2}(5)h_{1}^{\perp}(5)\right]_{1}Y_{2}(3)
\end{eqnarray}
 Hence, after five time slots, each receiver has three different  linear combinations of its three desired messages and the total achieved DoF for the 3-user BC is given by $D_{\Sigma}(3)= \frac{9}{5}$.
\begin{Theorem}
    The K-user broadcast channel with synergistic alternating CSIT with distribution $\in \Lambda (\lambda_{P} = \frac{(K-1)^{2}}{2K^{2}-K}, \lambda_{D} = \frac{K-1}{2K-1}, \lambda_{N}=\frac{1}{K})$ can achieve almost surely
    \begin{equation}
        D_{\Sigma}(K) = \frac{K^{2}}{2K-1}
    \end{equation}

\end{Theorem}

\begin{proof}
The transmission scheme starts with sending information symbols in
phase one, i.e., interference creation phase, to provide each
receiver with a linear combination of its intended data symbols
while creating $K-1$ interference terms at each receiver. This
phase consumes $K$ time slots to deliver $K$ different linear
combinations of the data symbols to $K$ different receivers while
creating $K \times (K-1)$ interference terms that will be useful
as a side information for the receivers in the subsequent time
slots. This phase requires $K \times (K-1)$ delayed CSIT states and
$K$ no CSIT states.

In contrast, phase two, i.e., interference resurrection phase,
consumes $(K-1)$ time slots to deliver $(K-1)$ messages of
order-$K$, i.e., intended for the $K$ receivers, in order to make
each receiver decode $K$ symbols successfully. This phase requires
$(K-1)^{2}$ perfect CSIT states and $(K-1)$ no CSIT states. The fraction of CSIT
states during the two phases is given by
\begin{eqnarray}
\lambda_{P}&=& \frac{(K-1)^{2}}{K \times (2K-1)}=\frac{(K-1)^{2}}{2K^{2}-K}\\
\lambda_{D} &=& \frac{K \times (K-1)}{K \times (2K-1)}=\frac{K-1}{2K-1}\\
\lambda_{N} &=& \frac{(2K-1)}{K \times (2K-1)}=\frac{1}{K}
\end{eqnarray}
%

\end{proof}


\section{Discussion}
\noindent Remark 1: \textbf{Comparison with all delayed CSIT
\cite{maddah2012completely}}\\ For the $K$-user BC model, the achievable
DoF under the CSIT alternation pattern with the 
distribution  given in Theorem 1  is
strictly greater than the best known upper bound for the all delayed
CSIT pattern \cite{maddah2012completely}, i.e., with distribution
$\Lambda(0,1,0)$, which is $K/({1+\frac{1}{2}+\dots+\frac{1}{K}})$
DoF. In order to send $K^{2}$ successfully decoded
messages, the proposed scheme in \cite{maddah2012completely} needs
$K \times (1+\frac{1}{2}+\dots+\frac{1}{K})\approx K \times
\ln(K)$ time slots while our proposed scheme needs only $2K-1$
time slots thanks to the alternating CSIT feature. Fig.
\ref{fig::BC_MAT_vs_Proposed} shows the synergistic benefits of
CSIT alternation on the DoF versus the number of users $K$.

\begin{figure}[ht!]
\centering
\includegraphics[scale=0.5]{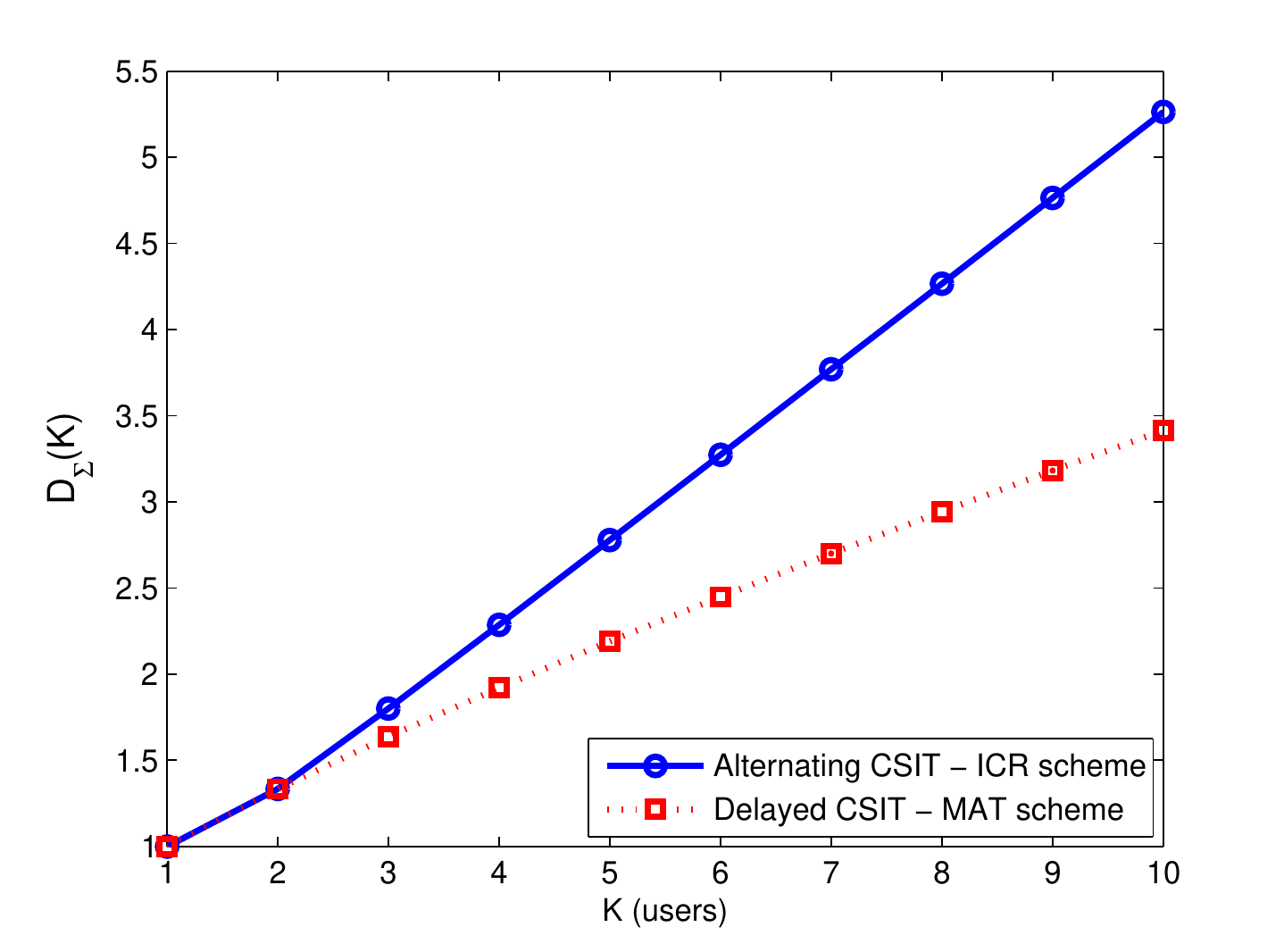}
\caption{DoF comparison for broadcast channel between all delayed and alternating CSIT models.}
\label{fig::BC_MAT_vs_Proposed}
\end{figure}

\noindent Remark 2: \textbf{Comparison with Hybrid CSIT
\cite{amuru2014degrees}}\\ The system model is similar to ours but
with hybrid CSIT, i.e., the link availability is constant over the
channel uses. As a comparison, for the case of $(P,D,D)$ the
proposed scheme in \cite{amuru2014degrees} achieves $\frac{9}{5}$
DoF, which implies that the channel states availability pattern
over the channel uses is $S_{123}^{5}=(PDD,PDD,PDD,PDD,PDD)$.
Note that this CSIT pattern  has a distribution given by
$\Lambda(\frac{5}{15},\frac{10}{15},\frac{0}{15})$. However, by
harnessing the synergy benefits of CSIT alternation in our case,
with less distribution of CSIT availability, i.e.,
$\Lambda(\frac{4}{15},\frac{6}{15},\frac{5}{15})$, our proposed
scheme can achieve the same $\frac{9}{5}$ DoF. Also, the authors
needed extensive channel extension to achieve this DoF by sending
18 symbols (10 symbols for $R_{1}$, 4 symbols for $R_{2}$ and 4
symbols for $R_{3}$) in 10 time slots. On the other hand, the proposed scheme requires only $5$ time slots to send 3 messages to each user.

\noindent Remark 3: \textbf{Comparison with Tandon et.al.
\cite{tandon2012minimum}}\\ For an M-antenna transmitter and K users, the proposed scheme in
\cite{tandon2012minimum}  assumes that at each time slot perfect
CSIT is present to $\min(M,K)$ receivers and no CSIT to the
remaining $K- \min(M,K)$ receivers. A total DoF of $\min(M,K)$ is
achievable at each time slot and therefore a sum DoF of
$\min(M,K)$  is also achievable for this scheme. The fraction of
time $\lambda$ that perfect CSIT obtained from any specific
receiver is $\min(M,K)/K$. For $M=K$, the fraction of time for
perfect CSIT $\lambda = 1$ which means perfect CSIT should be
available about all receivers.

\noindent Remark 4: \textbf{Synergy benefits of CSIT pattern}\\ The synergy
gain of delayed CSIT followed by perfect CSIT is useful to
reconstruct the interference terms in prior time slots and
constructing  messages useful for the receivers in subsequent time
slots. Note that the DoF for the 3-user BC with perfect CSIT is 3, with delayed CSIT is bounded by
$\frac{18}{11}$, and with no CSIT is one DoF. The alternation of
CSIT states $S_{123}$ over five time slots works cooperatively to provide a DoF greater than the DoF of the sum of their individual DoF for the same network. As an example, consider the CSIT alternation pattern given by $S^{5}_{123}=(NNN,DDD, DDD, DDD, PPP)$. If there is
no interaction between the five time slots, the DoF that can be obtained are given by $1 \times\frac{3}{15} + \frac{18}{11} \times \frac{9}{15} + 3 \times \frac{3}{15} = \frac{98}{55} < \frac{9}{5} $. However, harnessing the synergistic benefits of alternating 	CSIT, we can achieve more DoF ( $\frac{9}{5}$ DoF) with less CSIT pattern $S_{123}^{5} = (NDD,DND,DDN,PPN,PNP)$. Table 1 lists the beneficial synergistic CSIT alternation patterns with $\Lambda(\frac{4}{15},\frac{6}{15},\frac{5}{15})$ that can be utilized to achieve $\frac{9}{5}$ DoF for the 3-user BC channel. We can see from Table. 1 that there are only $|S_{123}^{5}|=|S_{123}^{3}| \times |S_{123}(4,5)|$ = 36  CSIT alternation patterns with synergistic benefits.


\begin{table}
\begin{center}
    \begin{tabular}{|c|c|}
    \hline
     $S_{123}^{3}$     & $S_{123}(4,5)$  \\[0.05in] \hline
    $(NDD,DND,DDN)$              & $(PPN, PNP)$       \\
    $(NDD,DDN,DND)$              & $(PNP, PPN)$        \\
    $(DND,DDN,DDN)$              & $(PPN, NPP)$        \\
   $ (DND,DDN,NDD)$ & $(NPP, PPN)$ \\
   $ (DDN,DND,NDD)$ & $(NPP, PNP)$ \\
      $(DDN,NDD,DND)$ & $(PNP, NPP)$ \\\hline
    \end{tabular}
    \vspace{0.2in}
    \caption{All synergistic CSIT patterns for the 3-user BC with $\Lambda(\frac{4}{15},\frac{6}{15},\frac{5}{15})$.}
    \end{center}
\end{table}

\noindent Remark 5: \textbf{Upper bound on the DoF }\\ An outer bound on the DoF region of the $K$-user BC under alternating CSIT was introduced in \cite{tandon2012minimum}. The achievable DoF, $\{d_i\}_{i=1}^K$, to the $K$ receivers is bounded by 
\begin{eqnarray}
K d_{1}+d_{2}+ \dots + d_{K} &\!\!\leq\!\! & K +(K-1)\gamma_{1}\\
d_{1}+K d_{2}+ \dots +d_{K} &\!\!\leq \!\!& K +(K-1)\gamma_{2}\\
\vdots \nonumber \\
d_{1}+d_{2}+ \dots +K d_{K} &\!\! \leq\!\! & K+(K-1)\gamma_{K}
\end{eqnarray}
\noindent where
\begin{equation}
    \gamma_{i}= \frac{\sum_{t=1}^{n} \mathds{I}(S_{i}(t)=P)}{n} \leq \gamma, \forall i=1, \dots, K
\end{equation}
\noindent is the fraction of time where perfect CSIT for receiver $i$ is available. Adding the the previous $K$ bounds, yields the following upper bound on the total DoF 
\begin{equation}
D_{\Sigma}(K) = d_{1}+d_{2}+ \dots +d_{K} \leq \frac{K^{2}+(K-1)\sum_{i=1}^{K}\gamma_{i}}{2K-1}
\end{equation}

Fig. \ref{fig::Upperbounds} depicts the comparison between the achievable DoF with perfect CSIT fraction $(\gamma_{1}, \gamma_{2}, \dots, \gamma_{K})$ where $\gamma_{i} = \gamma = \frac{K-1}{2K-1} $ and $\gamma_ {j \neq i} = \frac{K-2}{2K-1} < \gamma$, $\forall i,j \in\{1,\dots, K\}$ with the upper bound on the achievable DoF with the same alternating CSIT fraction, and the upper bound when $\gamma=1$ for the K-user BC.

\begin{figure}[ht!]
\centering
\includegraphics[scale=0.5]{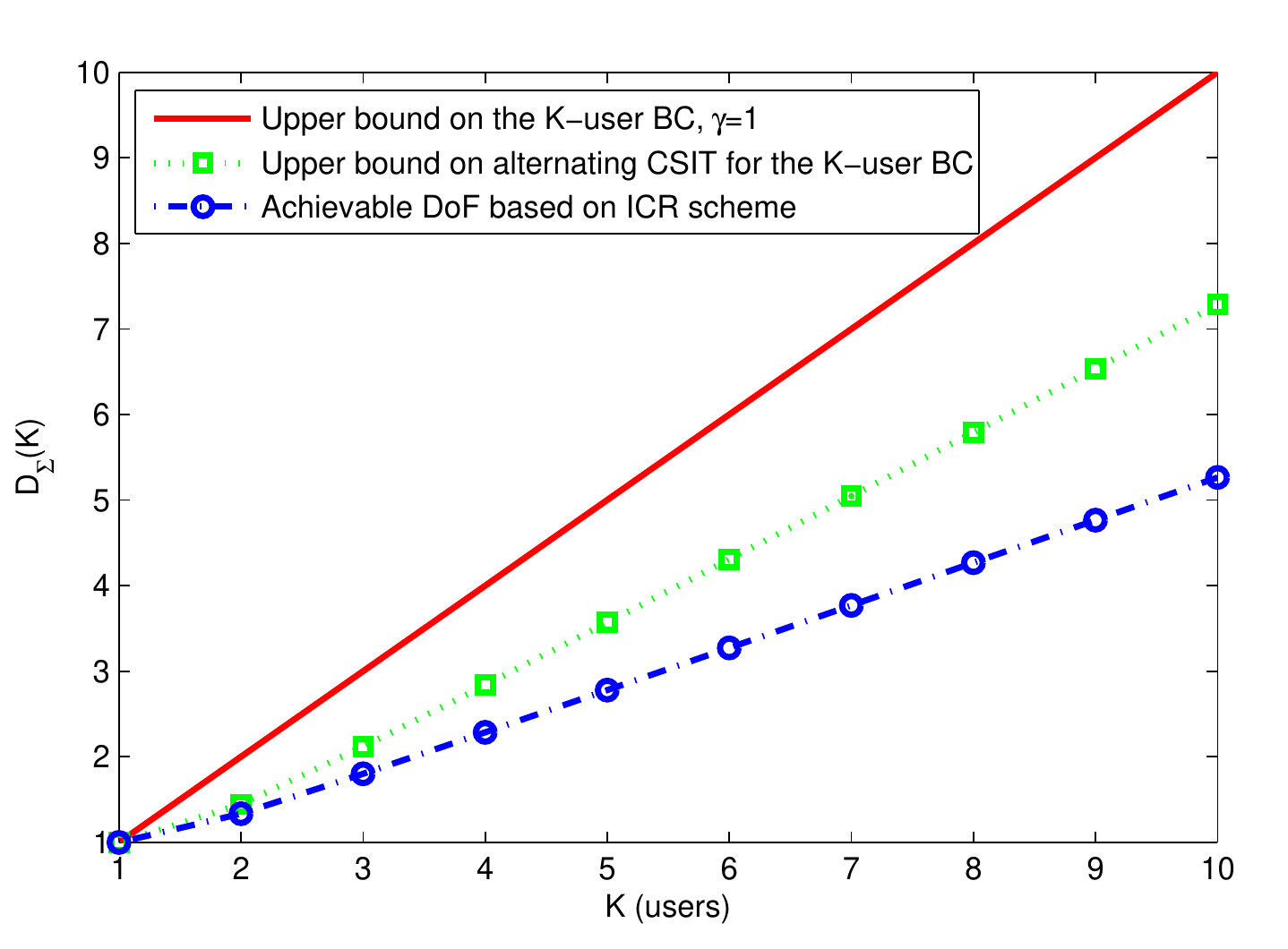}
\caption{DoF comparison for the K-user BC.}
\label{fig::Upperbounds}
\end{figure}

\noindent Remark 6: \textbf{DoF region characterization}

For the 3-user case, in order to find the optimal DoF for each receiver for a given perfect CSIT distribution ($\gamma_{1}$,$\gamma_{2}$,$\gamma_{3}$), we solve the following linear program
 \begin{eqnarray} 
\text{P1:}  \max_{d_{1},d_{2},d_{3}}  &\!\!\!\!\!\! & d_{1}+d_{2}+d_{3} \nonumber\\
 \text{s.t.}  &\!\!\!\!\!\! & 3 d_{1}+d_{2}+ d_{3}\leq  3 +2\gamma_{1} \\
 &\!\!\!\!\!\! & d_{1}+3 d_{2}+ d_{3} \leq  3 + 2\gamma_{2} \\
 &\!\!\!\!\!\! & d_{1}+d_{2}+ 3 d_{3}  \leq   3+2\gamma_{3}  \\
 &\!\!\!\!\!\! & 0 \leq d_{i} \leq  1, \hspace{0.1in} \forall i=1, 2, 3  
\end{eqnarray}

\noindent Since the constraints of the linear program are active, we can get a  general closed form expression as a function of $\gamma_{i}$'s by using the reduced echelon form method. Then, the solution will be as follows
\begin{equation}
d^{*}_{i} = \frac{3+4\gamma_{i}-\sum_{j=1, j\neq i}^{3}\gamma_{j}}{5}, \hspace{0.1in} \forall i=1,2,3
\end{equation}

For a perfect CSIT distribution $(\gamma_{1},\gamma_{2},\gamma_{3})=(\frac{2}{5},\frac{1}{5},\frac{1}{5})$ then the optimal DoF tuple is given by $d^{*}=(0.84,0.64,0.64)$ which is greater than the achievable DoF tuple $d = (0.6,0.6,0.6)$. Fig. \ref{fig::Region} shows the achievable DoF region for the 3-user BC: the red point is the achievable DoF under perfect CSIT fraction with $(\gamma_{1},\gamma_{2},\gamma_{3})=(2/5,1/5,1/5)$ (W.L.O.G we set $\gamma_{1}=\gamma$ and $\gamma_{i \neq 1} < \gamma$), and the time sharing scheme is achieved by any convex combinations of the corner points. 
\begin{figure}[ht!]
\centering
\includegraphics[scale=0.5]{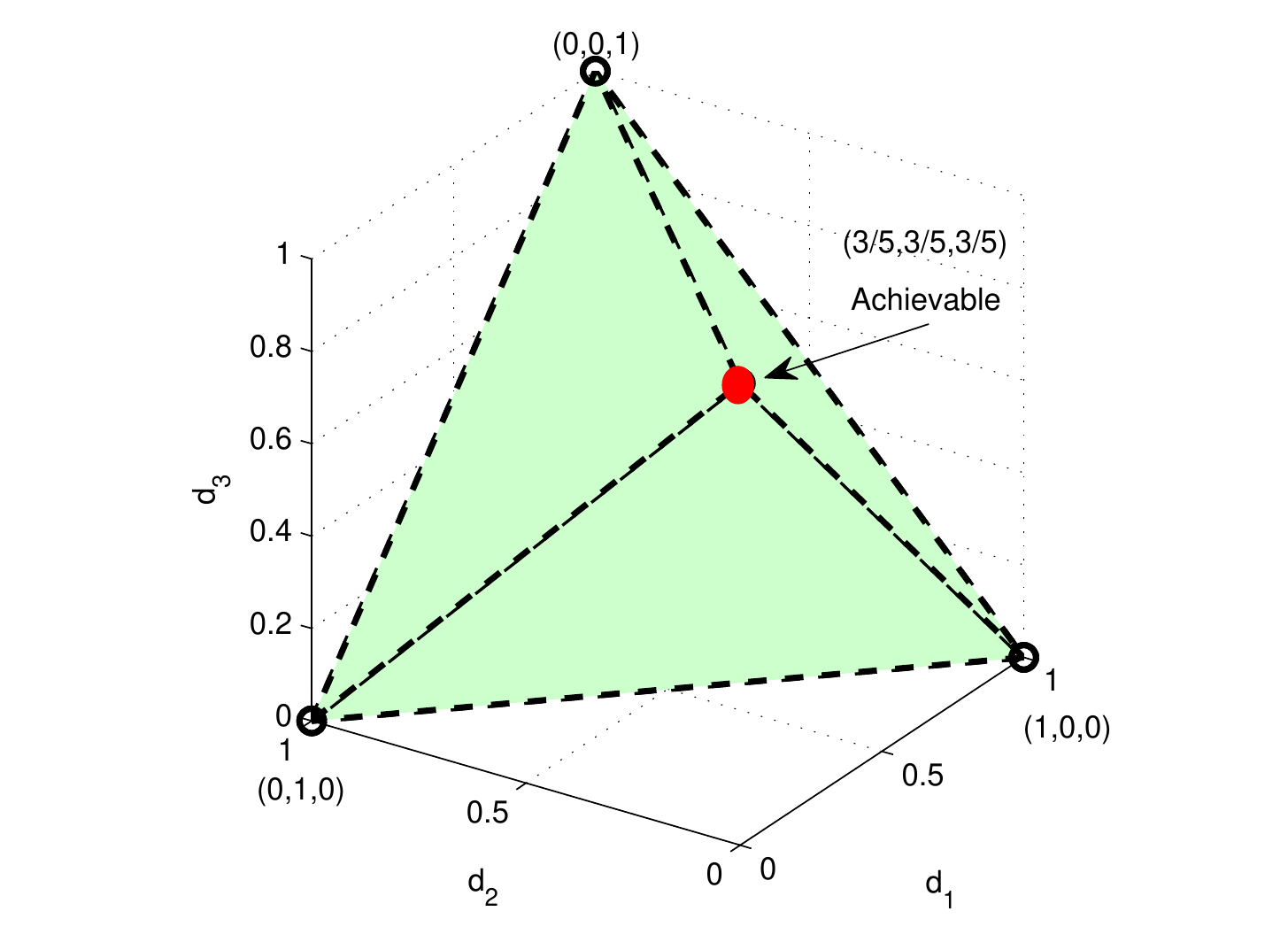}
\caption{\small{Achievable DoF Region for the 3-user BC.}}
\label{fig::Region}
\end{figure}

\begin{table}
\begin{center}
    \begin{tabular}{|c|c|c|}
    \hline
    $(\gamma_{1},\gamma_{2},\gamma_{3}) $      & $(d_{1},d_{2},d_{3})$ & Scheme \\ \hline
    $(1,0,0)$              & $(1,0,0)$ & ---        \\
    $(0,1,0)$              & $(0,1,0)$  & ---        \\
    $(0,0,1)$              & $(0,0,1)$ & ---          \\
   $(1/3,1/3,1/3)$ & $(1/3,1/3,1/3)$ & Time sharing \\
   $(2/5,1/5,1/5)$ & (3/5,3/5,3/5)& ICR \\
      $(1/5,2/5,1/5)$ & (3/5,3/5,3/5)& ICR\\
         $(1/5,1/5,2/5)$ & (3/5,3/5,3/5)& ICR \\
         $(1,1,1)$ & (1,1,1)& Conventional \\ \hline
    \end{tabular}
    \vspace{0.2in}
    \caption{\small{Perfect CSIT distribution among three users and its achievable degrees of freedom.}}
    \end{center}
\end{table}

\section{Conclusion}
We have investigated the synergistic benefits of the alternation
of CSIT for the K-user broadcast channel. The available CSIT
alternates between three possible states of availability
$(P,D,N)$. We have showed that $\frac{K^{2}}{2K-1}$  DoF can be
attained almost surely under CSIT distribution $\in \Lambda
(\lambda_{P} = \frac{(K-1)^{2}}{K \times (2K-1)}, \lambda_{D} =
\frac{K-1}{2K-1}, \lambda_{N}=\frac{1}{K})$. Also, we have
compared our scheme with prior work and highlighted the advantages
of having alternating CSIT to different receivers.

\nocite{*}
\bibliographystyle{IEEEtran}
\bibliography{mybibfile}

\end{document}